\UseRawInputEncoding

\documentclass{amsart}
\usepackage{mathrsfs}
\usepackage{graphicx}
\usepackage{amsmath}
\usepackage{amscd}
\usepackage{cite}
\usepackage{hyperref}
\usepackage{epsfig}
\usepackage{subfigure}
\usepackage{caption}
\usepackage{tikz}
  \usepackage{float}
  \usepackage{amssymb}
\usepackage{amsfonts}
\usepackage[all,cmtip]{xy}
\usepackage{appendix}
\usepackage{bm}

\newtheorem{Theorem}{Theorem}[section]
\newtheorem{Theorem/Definition}{Theorem/Definition}[section]
\newtheorem{Proposition}{Proposition}[section]
\newtheorem{Lemma}{Lemma}[section]
\newtheorem{Corollary}{Corollary}[section]

\newcommand{\pd}{\partial}

\newcommand{\bC}{{\mathbb C}}

\newcommand{\bP}{{\mathbb P}}

\newcommand{\bZ}{{\mathbb Z}}

\newcommand{\cB}{{\mathcal B}}
\newcommand{\cC}{{\mathcal C}}

\newcommand{\cF}{{\mathcal F}}

\newcommand{\cH}{{\mathcal H}}

\newcommand{\half}{\frac{1}{2}}

\newcommand{\wA}{{\widehat A}}

\newcommand{\te}{{\tilde e}}

\newcommand{\tpd}{{\tilde{\pd}}}

\newcommand{\hf}{{\hat f}}
\newcommand{\hF}{{\hat F}}

\newcommand{\be}{\begin{equation}}
\newcommand{\ee}{\end{equation}}
\newcommand{\bea}{\begin{eqnarray}}
\newcommand{\ben}{\begin{eqnarray*}}
\newcommand{\een}{\end{eqnarray*}}
\newcommand{\eea}{\end{eqnarray}}

\DeclareMathOperator{\Res}{Res}
\DeclareMathOperator{\Pf}{Pf}

\usepackage{color}

\definecolor{yellow}{rgb}{1,1,0}
\definecolor{orange}{rgb}{1,.7,0}
\definecolor{red}{rgb}{1,0,0}
\definecolor{green}{rgb}{0,1,1}
\definecolor{white}{rgb}{1,1,1}

\definecolor{A}{rgb}{.75,1,.75}

\theoremstyle{remark}
\newtheorem{Remark}{Remark}[section]

\begin{document}

\newtheorem{myDef}{Definition}
\newtheorem{thm}{Theorem}
\newtheorem{eqn}{equation}

\title[Kac-Schwarz Operators of Type $B$ and Quantum Spectral Curves]
{Kac-Schwarz Operators of Type $B$, Quantum Spectral Curves, and Spin Hurwitz Numbers}

\author{Ce Ji}
\address{School of Mathematical Sciences\\
Peking University\\Beijing, 100871, China}
\email{sms-jice@pku.edu.cn}

\author{Zhiyuan Wang}
\address{School of Mathematical Sciences\\
Peking University\\Beijing, 100871, China}
\email{zhiyuan19@math.pku.edu.cn}

\author{Chenglang Yang}
\address{Beijing International Center for Mathematical Research
 \& School of Mathematical Sciences\\
Peking University\\Beijing, 100871, China}
\email{yangcl@pku.edu.cn}

\begin{abstract}

Given a tau-function $\tau(\bm t)$ of the BKP hierarchy satisfying $\tau(0)=1$,
we discuss the relation between its BKP-affine coordinates on the isotropic Sato Grassmannian and its BKP-wave function.
Using this result,
we formulate a type of Kac-Schwarz operators for $\tau(\bm t)$
in terms of BKP-affine coordinates.
As an example,
we compute the affine coordinates of the BKP tau-function for spin single Hurwitz numbers
with completed cycles,
and find a pair of Kac-Schwarz operators $(P,Q)$ satisfying $[P,Q]=1$.
By doing this,
we obtain the quantum spectral curve for spin single Hurwitz numbers.

\end{abstract}

\maketitle


\section{Introduction}

In this paper,
we study the quantum spectral curve associated to a tau-function of the BKP hierarchy
using BKP-affine coordinates and a type of Kac-Schwarz operators,
and apply this strategy to construct the quantum spectral curve for
spin single Hurwitz numbers with completed cycles.

\subsection{KP tau-functions, BKP tau-functions, and Kac-Schwarz operators}
\label{sec-intro-BKP}

In Kyoto School's approach,
there are three equivalent descriptions of a tau-function of the KP hierarchy --
it corresponds to a point in the Sato Grassmannian,
or a vector in the bosonic or fermionic Fock space
satisfying the Hirota bilinear relations,
see e.g. \cite{djm, sa, sw}.
One has the following diagram:
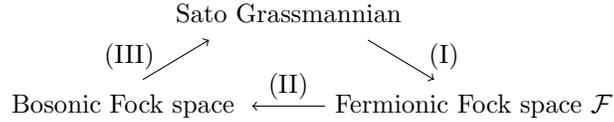
\begin{equation*}
\begin{tikzpicture}[scale=0.88]
\node [align=center,align=center] at (0,0) {Sato Grassmannian};
\node [align=center,align=center] at (2.8,-1.4) {Fermionic Fock space $\cF$};
\node [align=center,align=center] at (-2.5,-1.4) {Bosonic Fock space};
\draw [->] (1.2,-0.4) -- (2.2,-1);
\draw [<-] (-1.2,-0.4) -- (-2.2,-1);
\draw [<-] (-0.55,-1.4) -- (0.55,-1.4);
\node [align=center,align=center] at (2.35,-0.65) {(I)};
\node [align=center,align=center] at (-2.4,-0.65) {(III)};
\node [align=center,align=center] at (0,-1.1) {(II)};
\end{tikzpicture}
\end{equation*}
The arrow (I) is called the Pl\"ucker map,
and can be constructed via infinite wedges, see e.g. \cite[\S 14]{kac}.
The arrow (II) is the so-called boson-fermion correspondence,
see \cite[\S 5]{djm}.
One method to achieve (III) is to use partial derivatives of the wave function
associated to a tau-function (in the bosonic Fock space),
see e.g. \cite{av}.

When a KP tau-function $\tau(\bm t)$ lies in the big-cell $Gr_{(0)}$ of the Sato Grassmannian
(or equivalently, when $\tau(0) \not = 0$ is satisfied),
the above three descriptions of $\tau(\bm t)$ and the three arrows in the diagram
can be efficiently described in terms of its KP-affine coordinates \cite{zhou1, by}.
For details, see Zhou \cite[\S 3-\S 4]{zhou1}.
Moreover,
Zhou derived a formula for the connected bosonic $n$-point functions associated to
a KP tau-function $\tau(\bm t)$
in terms of its KP-affine coordinates in \cite[\S 5]{zhou1},
which enables us to compute the coefficients of the free energy $\log\tau(\bm t)$.

Now we consider the case of the BKP hierarchy.
Similar to the case of KP,
one also has three equivalent descriptions of a BKP tau-function:
\be
\label{eq-intro-BKPdiagram}
\begin{tikzpicture}[scale=0.88]
\node [align=center,align=center] at (0,0) {Isotropic Sato Grassmannian};
\node [align=center,align=center] at (2.8,-1.4) {Fermionic Fock space $\cF_B$};
\node [align=center,align=center] at (-2.5,-1.4) {Bosonic Fock space};
\draw [->] (1.2,-0.4) -- (2.2,-1);
\draw [<-] (-1.2,-0.4) -- (-2.2,-1);
\draw [<-] (-0.55,-1.4) -- (0.55,-1.4);
\node [align=center,align=center] at (2.35,-0.65) {(I)};
\node [align=center,align=center] at (-2.4,-0.65) {(III)};
\node [align=center,align=center] at (0,-1.1) {(II)};
\end{tikzpicture}
\ee
Here the Grassmannian corresponding to the BKP hierarchy is the isotropic Sato Grassmannian,
see e.g. \cite[\S 7]{hb} and  \cite[\S 4]{al3}.
In this paper,
we will use the construction in \cite[\S 7]{hb} of the isotropic Sato Grassmannian,
since in this construction the above arrows can be clearly described in terms of the BKP-affine coordinates
$\{a_{n,m}\}_{n,m\geq 0}$ (which satisfy the condition $a_{n,m} = -a_{m,n}$).
The arrow (I) is called the Cartan map,
see \cite[\S 7.3]{hb};
and the arrow (II) is the boson-fermion correspondence of type $B$ \cite{djkm, jm}.
See \cite{wy} for descriptions in terms of BKP-affine coordinates.
In particular,
a formula for connected $n$-point functions in terms of BKP-affine coordinates has been derived in \cite[\S 4]{wy}.
This is the type $B$ analogue of Zhou's formula in \cite{zhou1}.

Our first goal in this work is to describe (III) in terms of BKP-affine coordinates.
Inspired by the arrow (III) in the KP case,
we use the derivatives of wave functions to do this.
Let $\tau(\bm t)$ be a tau-function of the BKP hierarchy satisfying $\tau(0)=1$,
and let $w_B(\bm t;z)$ be the BKP-wave function defined by Sato's formula:
\begin{equation*}
w_B (\bm t;z) = \exp\Big(\sum_{k=0}^\infty t_{2k+1} z^{2k+1}\Big)\cdot
\frac{\tau (t_1 - \frac{2}{z}, t_3 -\frac{2}{3z^3}, t_5-\frac{2}{5z^5},\cdots)}{\tau (t_1,t_3,t_5,\cdots)}.
\end{equation*}
We prove the following (see Theorem \ref{thm-wave&affine}):
\begin{Theorem}
\label{thm-intro-wave}
The BKP-affine coordinates $\{a_{n,m}\}_{n,m\geq 0}$ of $\tau(\bm t)$ can be recovered from the wave function by the following relation:
\begin{equation*}
\begin{split}
\text{span} \big\{\pd_{t_1}^kw_B(0;z) \big\}_{k\geq 0} =&
\text{span} \big\{ z^k + a_{k,0}+\sum_{i=1}^\infty 2(-1)^{i} a_{k,i} z^{-i} \big\}_{k\geq 0}\\
=& \text{span} \big\{ z^k + \sum_{i=1}^\infty 2(-1)^{i} (a_{k,i} - a_{k,0}a_{0,i})
z^{-i} \big\}_{k\geq 0}
\end{split}
\end{equation*}
as sub-vector spaces of $\cH = \bC[z]\oplus z^{-1} \bC [[z^{-1}]]$.
\end{Theorem}

The basis vectors in the first line
are actually the fermionic $1$-point functions:
\begin{equation*}
\begin{split}
&\langle \phi_{0} \phi(z) e^A \rangle =
\half + \sum_{i=1}^\infty (-1)^{i} a_{0,i} z^{-i},\\
&\langle \phi_{-k} \phi(z) e^A \rangle = (-1)^k\cdot
\Big( z^k + a_{k,0}+\sum_{i=1}^\infty 2(-1)^{i} a_{k,i} z^{-i} \Big),
\qquad k\geq 1.
\end{split}
\end{equation*}

Inspired by this description of (III),
we propose a possible way to describe Kac-Schwarz operators \cite{ks, sc} for BKP tau-functions
in terms of BKP-wave functions and affine coordinates.
Denote by $U_\tau$ the sub-vector space of $\cH$ discussed in the above theorem,
then we say an operator $P$ is a Kac-Schwarz operator of type $B$ for $\tau(\bm t)$ if
$P(U_\tau) \subset U_\tau$ holds.

\subsection{Quantum spectral curves}
\label{sec-intro-QSC}

The above constructions inspire us to consider the
quantum spectral curve associated to a BKP tau-function.

Here we first briefly review Gukov-Su{\l}kowski's construction of quantum spectral curves \cite{gs}
via Eynard-Orantin topological recursion \cite{eo}.
Let $\cC$ be the plane curve
\begin{equation*}
\cC = \{A(u,v) = 0\}
\end{equation*}
on the complex $2$-space $\bC\times \bC$
equipped with the symplectic form $\omega=\frac{\sqrt{-1}}{\hbar}du\wedge dv$,
where $A$ is a function in $u$ and $v$.
Then $\cC$ is a Lagrangian submanifold of this complex $2$-space.
A quantization procedure is supposed to turn the coordinates
$u$ and $v$ into operators $\hat u$ and $\hat v$ respectively
which satisfy the commutation relation $[\hat v,\hat u]=\hbar$,
and the algebra of functions in $u,v$ into a noncommutative algebra of operators.
The quantization of the polynomial $A(u,v)$ is supposed to be an operator of the following form:
\begin{equation*}
\widehat A
=\widehat A_0 +\hbar \widehat A_1 +\hbar^2 \widehat A_2 +\cdots,
\end{equation*}
where $\widehat A_0$ is supposed to be the canonical quantization of $A$,
i.e.,
the operator obtained from $A$ by replacing $u,v$ by $\hat u,\hat v$ respectively.
Inspired by the matrix model origin of the Eynard-Orantin topological recursion,
Gukov and Su{\l}kowski defined the following Baker-Akhiezer function
(see \cite[\S 2]{gs}):
\begin{equation*}
\Psi(z):=\exp\Big(\sum_{n=0}^{\infty}\hbar^{n-1}S_n(z)\Big),
\end{equation*}
where
\be
\label{eq-def-Sn}
\begin{split}
&S_0(z):=\int^z v(z)du(z),\qquad
S_1(z):=-\half \log \frac{du}{dz},\\
&S_n(z):=\sum_{2g-1+k=n}\frac{1}{k!}\int^z\cdots\int^z
\omega_{g,k}(z_1,\cdots,z_k),
\quad n\geq 2,
\end{split}
\ee
and
$\omega_{g,k}(z_1,\cdots,z_k) = W_{g,k} (z_1,\cdots,z_k) dz_1\cdots dz_k$
are the Eynard-Orantin invariants associated to
the spectral curve $u=u(z),v=v(z)$ and a Bergman kernel $B(p,q)$.
The E-O invariants are symmetric muliti-linear differentials on the spectral curve.
The integrals in the above formulas are integrals of differentials along a path on the spectral curve,
and can be referred to as the anti-derivatives of these differentials.
For curves of genus greater than zero,
one needs to consider more general Baker-Akhiezer function
which includes non-perturbative corrections represented by certain $\theta$-functions,
see \cite[\S 2.1]{gs} for details.
Gukov-Su{\l}kowski conjectured that the quantum spectral curve $\widehat A$ associated to $A$
can be obtained by solving the Schr\"odinger equation:
\begin{equation*}
\widehat A \big( \Psi(z) \big) =0.
\end{equation*}
Moreover,
it is known that if a KP tau-function $\tau(\bm t)$ can be reconstructed from
E-O topological recursion,
then the quantum spectral curve can be understood using principal specialization.
In fact,
in this case the B-A function $\Psi(z)$ (up to some unstable terms)
coincide with the first basis vector of the corresponding point in the Sato Grassmannian,
which can be obtained from the tau-function $\tau(\bm t)$ by the principal specialization $t_k = -\frac{1}{kz^k}$.
See \cite{al1, zhou2, zhou4, ms, mss, zhou5, kz, al2, zhou6, dn} for examples of quantum spectral curves in this sense.

Now in this work,
we will understand the quantum spectral curve
associated to a BKP tau-function using the same point of view.
We call an operator $\Delta$
a quantum spectral curve for a BKP tau-function $\tau$
if
$\Delta (\Phi_0^B) = 0$,
where
\begin{equation*}
\Phi_0^B(z) = 1+\sum_{i=0}^\infty 2(-1)^i a_{0,i} z^{-i}
 = w_B(0;z)
\end{equation*}
is the first basis vector of $U_\tau$.
Similar to the case of KP,
given a BKP tau-function $\tau(\bm t)$,
in general we hope to find a pair of Kac-Schwarz operators $(P,Q)$ of type $B$
such that $P(\Phi_0^B)=0$,
and $Q(\Phi_k^B) \in U_\tau$ is a Laurent series of degree $k+1$.
Moreover, we may hope $[P,Q]=1$.
Here $P$ plays the role of a quantum spectral curve.
If $\tau(\bm t)$ can be reconstructed from E-O recursion
on a curve which is the semi-classical limit of $P$,
then $P$ is the quantum spectral curve in the sense of \cite{gs}.

\subsection{Spin single Hurwitz numbers}

As an example,
we derive such a pair of Kac-Schwarz operators of type $B$
and quantum spectral curve for spin single Hurwitz numbers with completed $(r+1)$-cycles \cite{eop, gkl},
We first compute the BKP-affine coordinates for the associated tau-function
using the fermionic representation in \cite{gkl, le}.
The result is (see \S \ref{sec-sH-affine}):
\begin{equation*}
\begin{split}
&a_{0,n}^{(r,\vartheta)} = -a_{n,0}^{(r,\vartheta)} =
  \frac{p^n}{ 2\cdot n!}\cdot \exp\big( \beta\frac{n^{r+1}}{r+1} \big),
\qquad \forall n>0;\\
&a_{n,m}^{(r,\vartheta)} = \frac{p^{m+n}}{4\cdot m!\cdot n!} \cdot \frac{m-n}{m+n} \cdot
\exp \big( \beta\frac{m^{r+1} + n^{r+1}}{r+1} \big),
\qquad \forall m,n>0;
\end{split}
\end{equation*}
and $a_{0,n}^{(r,\vartheta)} = 0$.
Then we derive the Kac-Schwarz operators using the above explicit expressions.
We have (see \S \ref{sec-sH-KS}):
\begin{Theorem}
Let $P,Q$ be the following operators:
\begin{equation*}
\begin{split}
P= &\exp\Big( \frac{\beta}{r+1} \cdot \sum_{i=0}^{r}z^{-1}
\big(z\frac{\partial}{\partial z}\big)^iz\big(z\frac{\partial}{\partial z}\big)^{r-i}\Big)\partial_z\\
& \quad -p\exp\Big(\frac{2\beta}{r+1}\cdot
\sum_{i=0}^{r}z^{-2}\big(z\frac{\partial}{\partial z}\big)^iz^{2}
\big(z\frac{\partial}{\partial z}\big)^{r-i}\Big)z^{-2};\\
Q= &\exp\Big(- \frac{\beta}{r+1}\cdot \sum_{i=0}^{r}z
\big(z\frac{\partial}{\partial z}\big)^iz^{-1}\big(z\frac{\partial}{\partial z}\big)^{r-i}\Big)z,
\end{split}
\end{equation*}
then $P$ and $Q$ are Kac-Schwarz operators of the BKP tau-function $\tau^{(r,\vartheta)}$
for spin single Hurwitz numbers with $(r+1)$-completed cycles,
satisfying $[P,Q]=1$.
Moreover,
we have
$P(\Phi_0^{(r,\vartheta)}) =0$,
and
\begin{equation*}
Q(\Phi_k^{(r,\vartheta)}) =
e^{f^{(r,\vartheta)}(k)-f^{(r,\vartheta)}(k+1)}\Phi_{k+1}^{(r,\vartheta)}
-\frac{p^{k+1} \cdot e^{f^{(r,\vartheta)}(k)}}{(k+1)!}\Phi_0^{(r,\vartheta)},
\qquad \forall k\geq 0.
\end{equation*}
\end{Theorem}

By taking the semi-classical limit of $P$,
we obtain a plane curve
on the complex $2$-space $\bC^* \times \bC$ with sympletic coordinates $(u,v)$
where $x=e^u$ and $y=v$
(see \S \ref{sec-sH-QSC-main} for details):
\begin{equation*}
x=-ye^{-y^r}.
\end{equation*}
It has been conjectured by Giacchetto, Kramer, Lewa\'nski \cite{gkl}
and proved by Alexandrov and Shadrin \cite{as} that
the spin single Hurwitz numbers can be reconstructed from
the E-O topological recursion on this curve,
and thus $P$ is indeed the quantum spectral curve in the sense of Gukov-Su{\l}kowski.
It is worth mentioning that the Bergman kernel used in \cite{gkl, as}
coincides with the correction term in the formula \cite[(74)]{wy} for connected $2$-point function
of a BKP tau-function.

The rest of this paper is arranged as follows.
In \S \ref{sec-pre},
we recall some basics of the boson-fermion correspondence of type $B$.
In \S \ref{sec-isotr-Gr},
we first review the isotropic Sato Grassmannian following \cite{hb},
and then prove Theorem \ref{thm-intro-wave}.
This leads to the notion of Kac-Schwarz operators of type $B$.
In \S \ref{sec-sH-affine},
we compute the BKP-affine coordinates of the tau-function for
spin single Hurwitz numbers with completed cycles.
Finally in \S \ref{sec-sH-KS},
we use the explicit expressions of the affine coordinates
to derive the K-S operators $(P,Q)$ and
quantum spectral curve for the spin Hurwitz numbers.

\section{Preliminaries of Boson-Fermion Correspondence of Type $B$}
\label{sec-pre}

In this section,
we briefly review some preliminaries of the boson-fermion correspondence of type $B$.
We first recall the construction of the fermionic Fock space and boson-fermion correspondence
via neutral fermions \cite{yo, djkm, jm, va},
and then recall the notion of BKP tau-functions and BKP-affine coordinates \cite{hb, wy}.

\subsection{Neutral fermions and fermionic Fock space of type $B$}

First we recall the construction of the fermionic Fock space of type $B$
via neutral fermions.
See \cite{djkm, jm} for details.

The neutral fermions are a family of operators $\{\phi_m\}_{m\in \bZ}$
satisfying the anti-commutation relations:
\be
\label{eq-anticomm}
[\phi_m, \phi_n]_+ =(-1)^m \delta_{m+n,0},
\ee
where $[\phi_m, \phi_n]_+ = \phi_m\phi_n + \phi_n\phi_m$.
In particular, $\phi_0^2=\half$
and $\phi_n^2=0$ for $n\not=0$.
The fermionic Fock space $\cF_B$ of type $B$ is the infinite-dimensional $\bC$-vector space
of all formal (infinite) summations
\begin{equation*}
\sum
c_{k_1,\cdots,k_n}
\phi_{k_1} \phi_{k_2} \cdots \phi_{k_n} |0\rangle,
\qquad
c_{k_1,\cdots,k_n} \in\bC,
\end{equation*}
over $n\geq 0$ and $k_1>\cdots >k_n \geq 0$,
where $|0\rangle$ is a vector satisfying:
\be
\label{eq-B-anni}
\phi_i |0\rangle = 0,
\qquad \forall i <0.
\ee
The vector $|0\rangle$ is called the fermionic vacuum vector.
The operators $\{\phi_n\}_{n\geq 0}$ are called the fermionic creators,
and $\{\phi_n\}_{n<0}$ are called the fermionic annihilators.
The Fock space $\cF_B$ can be decomposed as a direct sum of even and odd parts:
\begin{equation*}
\cF_B = \cF_B^0 \oplus \cF_B^1,
\end{equation*}
where $\cF_B^0$ and $\cF_B^1$ are the subspaces with
even and odd numbers of the creators $\{\phi_i\}_{i\geq 0}$ respectively.

The even part $\cF_B^0$ has a natural basis
indexed by all strict partitions.
(Recall that a partition of integer $\mu=(\mu_1,\cdots,\mu_n)$ is called strict
if $\mu_1>\mu_2\cdots>\mu_n>0$.)
Denoted by $DP$ the set of all strict partitions,
and here we allow the empty partition $(\emptyset) \in DP$.
Given a strict partition $\mu= (\mu_1>\cdots >\mu_n > 0) $,
denote:
\be
|\mu\rangle
= \begin{cases}
\phi_{\mu_1}\phi_{\mu_2}\cdots \phi_{\mu_n}|0\rangle, &\text{ for $n$ even;}\\
\sqrt{2}\cdot \phi_{\mu_1}\phi_{\mu_2}\cdots \phi_{\mu_n} \phi_0 |0\rangle, &\text{ for $n$ odd}.
\end{cases}
\ee
In particular, one has $|(\emptyset)\rangle =|0\rangle$.
Then $\{|\mu\rangle\}_{\mu\in DP}$ form a basis for $\cF_B^0$.

The dual Fock space $\cF_B^*$ is defined to be the vector space spanned by:
\begin{equation*}
\langle 0|
\phi_{k_n}  \cdots \phi_{k_2} \phi_{k_1},
\qquad
k_1 < k_2 < \cdots < k_n \leq 0,
\quad n\geq 0,
\end{equation*}
where $\langle 0|$ is a vector satisfying:
\be
\label{eq-B-anni-2}
\langle 0| \phi_i = 0,
\qquad \forall i >0.
\ee
Then there is a nondegenerate pairing $\cF_B^* \times \cF_B \to \bC$ determined by
\eqref{eq-anticomm}, \eqref{eq-B-anni}, \eqref{eq-B-anni-2},
and the requirements
$\langle 0 | 0 \rangle =1$ and $\langle 0|\phi_0 | 0\rangle =0$.
It is easy to check that for an arbitrary sequence $k_1>k_2>\cdots >k_n\geq 0$,
\be
\label{eq-basic-VEV}
\langle 0 | \phi_{-k_n}\cdots \phi_{-k_1}
\phi_{k_1}\cdots \phi_{k_n} |0\rangle = \begin{cases}
(-1)^{k_1+\cdots +k_n}, &\text{ if $k_n\not=0$;}\\
\half\cdot (-1)^{k_1+\cdots +k_{n-1}},&\text{ if $k_n=0$.}
\end{cases}
\ee
In general,
the vacuum expectation value of a product of neutral fermions
can be computed using Wick's Theorem:
\begin{equation*}
\langle 0| \phi_{i_1}\phi_{i_2}\cdots \phi_{i_{2n}}|0\rangle
=\sum_{\substack{(p_1,q_1,\cdots,p_n,q_n)\\p_k<q_k, \quad p_1<\cdots<p_n}}
\text{sgn}(p,q)\cdot \prod_{j=1}^n
\langle 0| \phi_{i_{p_j}}\phi_{i_{q_j}} |0\rangle,
\end{equation*}
where $(p_1,q_1,\cdots,p_n,q_n)$ is a permutation of $(1,2,\cdots,2n)$,
and $\text{sgn}(p,q)$ denotes its sign ($\text{sgn}=1$ for an even permutation, and $\text{sgn}=-1$ for an odd one).
In what follows,
we will denote by
\begin{equation*}
\langle A \rangle = \langle 0 |A|0\rangle
\end{equation*}
the vacuum expectation value of an operator $A$ on the fermionic Fock space.

The normal-ordered product $:\phi_i\phi_j:$ of two neutral fermions
is defined by:
\be
\label{eq-def-normal-order}
:\phi_i\phi_j: = \phi_i\phi_j -\langle \phi_i\phi_j \rangle.
\ee
In particular,
one has $:\phi_0^2: = \phi_0^2 -\langle \phi_0^2\rangle=0$.
The relation \eqref{eq-anticomm} is equivalent to:
\be
\label{eq-normal-fermfield}
\phi(w)\phi(z)
= :\phi(w)\phi(z): +i_{w,z}\frac{w-z}{2(w+z)},
\ee
where $\phi(z)$ is the fermionic field:
\be
\label{eq-ferm-field}
\phi(z) = \sum_{i\in \bZ} \phi_i z^i,
\ee
and $i_{w,z}$ means formally expanding on $\{|w|>|z|\}$, i.e.,
\be
i_{w,z}\frac{w-z}{2(w+z)}= \half +
\sum_{j=1}^\infty (-1)^{j} w^{-j} z^j.
\ee

\subsection{Boson-fermion correspondence of type $B$}
\label{sec-pre-bf}

In this subsection,
we recall the boson-fermion correspondence of type $B$.
See \cite{djkm, jm}.

First we recall the construction of the bosonic operators $\{H_n\}_{n:\text{ odd}}$ of type $B$.
Let $n\in 2\bZ+1$ be an odd integer,
and define $H_n$ by:
\be
\label{eq-def-boson}
H_n = \half \sum_{i\in \bZ} (-1)^{i+1} \phi_i\phi_{-i-n}.
\ee
Then they satisfy the following canonical commutation relation:
\begin{equation*}
[H_n,H_m] = H_nH_m-H_mH_n= \frac{n}{2}\cdot \delta_{m+n,0},
\qquad
\forall n,m \text{ odd}.
\end{equation*}
The operators $H_n$ are called the bosons of type $B$.
One has:
$H_n |0\rangle = 0$,
for every $n> 0$ odd.
Denote by $H(z)$ the generating series of these bosons:
\be
H(z)= \sum_{n\in \bZ: \text{ odd}} H_n z^{-n}
\ee
Notice that one can also define $H_{2k}$ using \eqref{eq-def-boson},
and the anti-commutation relation \eqref{eq-anticomm} implies
$H_{2k}=0$ for every $k\not= 0$.
In this sense,
one easily checks that:
\be
H(z) = -\half :\phi(-z)\phi(z):.
\ee

Now we recall the boson-fermion correspondence of type $B$.
Let $\bm t=(t_1,t_3,t_5,\cdots)$ be a family of formal variables,
and denote:
\be
\label{eq-def-Gamma+-}
\begin{split}
&\Gamma_+^B(\bm t) = \exp\Big( \sum_{n>0: \text{ odd}} t_n H_n \Big)
= \exp \Big(
\half \sum_{n>0: \text{ odd}} t_n \sum_{i\in \bZ} (-1)^{i+1} \phi_i\phi_{-i-n}\Big),\\
&\Gamma_-^B(\bm t) = \exp\Big(
\sum_{n>0: \text{ odd}} t_n H_{-n} \Big)
= \exp\Big( \half \sum_{n>0: \text{ odd}} t_n \sum_{i\in \bZ} (-1)^{i+1} \phi_i\phi_{-i+n}\Big).
\end{split}
\ee
\begin{Theorem}
[\cite{djkm}]
The following map is a linear isomorphism:
\begin{equation*}
\sigma_B :\cF_B \to \bC[w][\![t_1,t_2,\cdots]\!]/\sim,
\qquad
|U\rangle \mapsto \sum_{i=0}^1 \omega^i\cdot \langle i |\Gamma_+^B(\bm t)|U\rangle,
\end{equation*}
where $\omega^2\sim 1$,
and $\langle 1| =\sqrt{2}\langle 0|\phi_0 \in (\cF_B^1)^*$.
Under this isomorphism,
one has:
\be
\label{eq-bfcor-boson}
\sigma_B (H_n |U\rangle) = \frac{\pd}{\pd t_n} \sigma_B(|U\rangle),
\qquad
\sigma_B (H_{-n}|U\rangle)=\frac{n}{2} t_n \cdot \sigma_B(|U\rangle),
\ee
for every odd $n>0$.
Moreover,
\be
\label{eq-bfcor-ferm}
\sigma_B (\phi(z)|U\rangle) =
\frac{1}{\sqrt{2}} \omega\cdot e^{\xi(\bm t,z)} e^{-\xi (\tilde\pd ,z^{-1})} \sigma_B(|U\rangle),
\ee
where
$\xi(\bm t,z) = \sum_{n>0\text{ odd}} t_{n}z^{n}$
and $\tilde\pd = (2\pd_{t_{ 1}},\frac{2}{3}\pd_{t_{ 3}},\frac{2}{5}\pd_{t_{ 5}},\cdots)$.
\end{Theorem}

Furthermore, one has the following:
\begin{Theorem}
[\cite{yo}]
\label{eq-bf-schur}
Let $\lambda = (\lambda_1>\cdots>\lambda_l>0)\in DP$, then one has:
\be
\label{eq-bf-schurq}
Q_\lambda(\bm t/2) = 2^{\half l(\lambda)}\cdot \sigma_B (\phi_{\lambda_1}\cdots \phi_{\lambda_l}|\alpha(\lambda)\rangle),
\ee
where $Q_\lambda$ is the Schur $Q$-function indexed by $\lambda\in DP$,
and
\begin{equation*}
|\alpha(\lambda)\rangle = \begin{cases}
|0\rangle, & \text{if $l(\lambda)$ is even;}\\
\sqrt{2}\phi_0|0\rangle, & \text{if $l(\lambda)$ is odd}.
\end{cases}
\end{equation*}
\end{Theorem}

See e.g. \cite{hh, sch, mac} for an introduction to Schur $Q$-functions and the relation to
the characters of projective representations of the symmetric groups $S_n$.

\subsection{Tau-Functions for BKP Hierarchy and Affine Coordinates}
\label{sec-pre-affine}

In this subsection,
we recall the tau-functions of the BKP hierarchy and their affine coordinates.
See \cite{hb, wy} for details.

BKP hierarchy was introduced by Kyoto School \cite{djkm}.
A tau-function of the BKP hierarchy
is the image of a vector $e^g|0\rangle \in \cF_0$ under boson-fermion correspondence:
\begin{equation*}
\tau(\bm t) = \langle 0| \Gamma_+^B(\bm t) e^g |0\rangle,
\end{equation*}
where $\bm t = (t_1,t_3,t_5,\cdots)$,
and $g$ is of the form
$g=\sum_{m,n\in \bZ} c_{m,n}:\phi_m\phi_n:$
where
\be
\label{eq-constraints-goinfty}
c_{m,n}=0,\qquad \text{ for $|m-n|>>0$}.
\ee
In this paper we will always regard a tau-function $\tau(\bm t)$ as
a formal power series in the time variables $t_1,t_3,t_5,\cdots$.
A function $\tau =\tau(\bm t)$ is a BKP tau-function if and only if
the following Hirota bilinear equation holds
(see \cite{djkm}):
\be
\label{eq-BKP-Hirota}
\Res_{z=0} \Big(
z^{-1} e^{\xi (\bm t-\bm s,z)}
\tau (\bm t -2[z^{-1}]) \tau(\bm s + 2[z^{-1}])
dz \Big)  = \tau(\bm t) \tau(\bm s),
\quad \forall \bm t, \bm s,
\ee
where $[z^{-1}]$ denotes the sequence $(\frac{1}{z},\frac{1}{3z^3},\frac{1}{5z^5},\cdots)$.

In the rest of this section,
we recall the BKP-affine coordinates for a tau-function
(in the big cell of the isotropic Sato Grassmanian).
Let $\tau = \tau(\bm t)$ be a tau-function of the BKP hierarchy
with initial value $\tau(\bm 0) = 1$,
then it can be uniquely represented as a Bogoliubov transform of the following form
in the fermionic Fock space:
\be
\tau(\bm t) = \sigma_B (e^A |0\rangle),
\ee
where $A$ is of the following form:
\be
\label{eq-bog-A}
A= \sum_{n,m\geq 0} a_{n,m} \phi_m \phi_n,
\qquad a_{n,m}\in \bC.
\ee
Notice that $A$ only involves fermionic creators,
and here we do not need to impose constraints such like \eqref{eq-constraints-goinfty} on the coefficients.
We will always assume:
\be
\label{eq-bog-A-cond}
a_{n,m} = -a_{m,n},
\qquad \forall n,m\geq 0,
\ee
since for $m,n\geq 0$ we have $\phi_m \phi_n = -\phi_n\phi_m$
unless $n=m=0$.
In particular,
$a_{n,n}=0$ for every $n\geq 0$.
The coefficients $\{a_{n,m}\}_{n,m\geq 0}$ are called the affine coordinates
of this BKP tau-function $\tau(\bm t)$.

By expanding the exponential $e^A$ and applying \eqref{eq-bf-schurq},
one easily sees that the tau-function $\tau$
admits an expansion by Schur $Q$-functions,
such that the coefficients are Pfaffians of its affine coordinates $\{a_{n,m}\}_{n,m\geq 0}$:
\be
\tau_A =
\sum_{\mu \in DP} (-1)^{\lceil l(\mu)/2 \rceil} \cdot
\Pf (a_{\mu_i,\mu_j})_{1\leq i,j\leq 2 \lceil l(\mu)/2 \rceil } \cdot Q_\mu(\bm t/2),
\ee
where $ l(\mu) $ is the length of
$\mu = (\mu_1> \cdots>\mu_{l(\mu)}>0) \in DP$,
and $\lceil \cdot \rceil$ is the ceiling function.
Here for $\mu\in DP$ with $l(\mu)$ odd,
we use the convention $\mu_{l(\mu)+1} = 0$.
For example, the first a few terms of $\tau_A$ are:
\begin{equation*}
\begin{split}
\tau_A =& 1+ \sum_{n>0}a_{0,n}\cdot Q_{(n)}(\bm t/2)
+\sum_{m>n>0} a_{n,m}\cdot Q_{(m,n)}(\bm t/2)\\
& + \sum_{m>n>l>0} (a_{n,m}a_{0,l} -a_{l,m}a_{0,n}+ a_{0,m}a_{l,n})Q_{(m,n,l)} (\bm t/2)\\
& +\sum_{m>n>l>k>0} (a_{n,m}a_{k,l} -a_{l,m}a_{k,n}+ a_{k,m}a_{l,n})Q_{(m,n,l,k)} (\bm t/2) +\cdots.
\end{split}
\end{equation*}

For an operator $A$ of the above form,
the following identity will be useful
(which can  be proved using the Baker-Campbell-Hausdorff formula,
see e.g. \cite[\S 7.3.4]{hb}):
\be
\label{eq-lem-conj}
e^{-A} \phi_i e^A
= \begin{cases}
\phi_i, & \text{ $i>0$;}\\
\phi_i+\sum_{m\geq 0} 2(-1)^i
(-a_{-i,m}+a_{-i,0} a_{0,m})\phi_m, & \text{ $i\leq 0$}.
\end{cases}
\ee

\section{Isotropic Sato Grassmannian and Kac-Schwarz Operators}
\label{sec-isotr-Gr}

It is known that there is a one-to-one correspondence between
the set of all tau-functions of BKP hierarchy (up to a constant)
and the isotropic Sato Grassmannian.
Similar to the case of KP hierarchy,
we have three equivalent descriptions for a BKP tau-function,
see the diagram \eqref{eq-intro-BKPdiagram} in the Introduction.
The arrow (II) is just the boson-fermion correspondence $\sigma_B$.
Now in this section,
we first review the construction of the isotropic Sato Grassmannian
and the Cartan map (I),
following \cite[\S 7]{hb}.
Then we describe a way to achieve (III) for elements in the big cell
of the isotropic Sato Grassmannian
(i.e., BKP tau-functions $\tau(\bm t)$ with  $\tau(0)=1$),
by recovering the BKP-affine coordinates from the wave function.

Inspired by this description of (III),
we propose a way to formulate the notion of Kac-Schwarz operators
for BKP tau-functions in terms of the associated BKP-wave function and BKP-affine coordinates.
We will also discuss the construction of quantum spectral curves
via principal specialization of BKP tau-functions.

\subsection{Isotropic Sato Grassmannian and affine coordinates}

In this subsection we recall some basics of the isotropic Sato Grassmannian.
See \cite[\S 7.3]{hb} for details.

Let $z$ be a formal variable,
and denote $\cH = \cH_+ \oplus \cH_-$,
where
\begin{equation*}
\cH_+ = \bC[z],
\qquad\qquad
\cH_- = z^{-1} \bC[[z^{-1}]].
\end{equation*}
Then $\{e_j = z^{-j-1}\}_{j\in \bZ}$ form a basis for $\cH$.
Let $\{\te^i\}_{i\in \bZ}$ be the dual basis for $\cH^*$,
and let $\cH_\phi \subset \cH\oplus \cH^*$ be the linear subspace
\be
\cH_\phi = \text{span}\{ e_i^0 \}_{i\in \bZ},
\ee
where
\be
e_i^0 = \frac{1}{\sqrt{2}} \big(e_i + (-1)^i \te^{-i} \big).
\ee
Let $Q_\phi : \cH_\phi \times \cH_\phi \to\bC$ be the nondegenerate symmetric bilinear form on $\cH_\phi$
given by $Q_\phi (e_i^0,e_j^0) = (-1)^{i+j} \cdot \delta_{i+j,0}$,
then the linear subspace
\be
\cH_\phi^0 = \text{span} \{e_i^0\}_{i<0} \subset \cH_\phi
\ee
is maximally isotropic with respect to $Q_\phi$.

Now let $O(\cH_\phi,Q_\phi) \subset GL(\cH_\phi)$ be the orthogonal subgroup:
\be
O(\cH_\phi,Q_\phi) = \big\{ g\in GL(\cH_\phi) \big|
Q_\phi(gu,gv) = Q_\phi(u,v) ,\forall u,v \in \cH_\phi \big\}.
\ee
Then the isotropic Sato Grassmannian $Gr_{\cH_\phi^0}^0 (\cH_\phi)$ is the orbit of $\cH_\phi^0$
under the action of $O(\cH_\phi,Q_\phi)$.
Then for every $w^0 = \text{span}\{w_i\cdots\}_{i>0} \in Gr_{\cH_\phi^0}^0 (\cH_\phi)$,
one has $Q_\phi (w_i,w_j) = 0$,
and there exists an element $g_\phi \in O(\cH_\phi, Q_\phi)$ such that
\be
w_i = g_\phi (e_{-i}^0),
\qquad \forall i>0.
\ee

Given an element $w^0 \in Gr_{\cH_\phi^0}^0 (\cH_\phi)$ in the isotropic Sato Grassmannian,
one can associate an element $Ca(w^0)$ in the projectivization $\bP (\cF_B)$
of the fermionic Fock space in the following way.
For a vector
\be
w = \sum_{i\in \bZ} c_i e_i^0 \in \cH_\phi,
\ee
denote
\be
\Phi(w) = \sum_{ i \in \bZ } c_i \phi_i.
\ee
Then for any $w^0 = \text{span}\{w_i\}_{i>0} \in Gr_{\cH_\phi^0}^0 (\cH_\phi)$£¬
define
(see \cite[(7.3.43)]{hb} for details):
\be
Ca(w^0) = [ \prod_{i= 1}^\infty \Phi(w_i) |0\rangle ],
\ee
where $[\cdot]$ denotes the equivalence class under projectivization.
This gives a map
\be
Ca: \quad Gr_{\cH_\phi^0}^0 (\cH_\phi) \to \bP(\cF_B),
\quad w^0 \mapsto Ca(w^0),
\ee
which is called the Cartan map.
It is the B-type analogue of the Pl\"ucker map (given by the infinite wedge)
on the ordinary Sato Grassmannian.
It is also the infinite-dimensional version of the map introduced by Cartan in \cite{ca}.
The image of this Cartan map in $\bP(\cF_B)$ is the intersection of an infinite number of quadrics
whose defining relations are called the Cartan relations,
which are equivalent to the BKP Hirota bilinear relation \eqref{eq-BKP-Hirota}.

When the element $w^0$ is in the big cell of the isotropic Sato Grassmannian,
the Cartan map can be represented explicitly in terms of the BKP-affine coordinates.
In fact,
let $\{a_{n,m}\}_{n,m\geq 0}$ be a family of complex numbers satisfying $a_{n,m} = -a_{m,n}$,
and let $w^0 = \text{span}\{w_k\}_{k>0} \in Gr_{\cH_\phi^0}^0 (\cH_\phi)$ where:
\be
\label{eq-def-basiswk}
\begin{split}
w_k =& e_{-k}^0 + 2 (-1)^{k+1}
\sum_{j=0}^\infty ( a_{j,k} - a_{0,k}a_{j,0}) e_j^0.
\end{split}
\ee
Then one has (see \cite[Theorem 7.3.1]{hb}):
\begin{Theorem}
[\cite{hb}]
Let $\{w_k\}_{k>0}$ be given by \eqref{eq-def-basiswk},
and let $w^0 = \text{span}\{w_k\}_{k>0}$ be a point in $Gr_{\cH_\phi^0}^0 (\cH_\phi)$.
Then the image of $w^0 $ under the Cartan map is:
\be
Ca(w^0) = \Big[ \exp \Big( \sum_{n,m\geq 0} a_{n,m} \phi_m\phi_n \Big) |0\rangle \Big].
\ee
\end{Theorem}
The isotropy condition for $w^0$ is equivalent to the relation $a_{n,m}=-a_{m,n}$.

\begin{Remark}
Here our notation $\{a_{n,m}\}$ and the notation $\{A_{n,m}\}$ in \cite[\S 7]{hb}
are related by $a_{n,m} = -\half A_{n,m}$.
\end{Remark}

\subsection{Recovering BKP-affine coordinates from wave function}
\label{sec-wave-to-affine}

Let $\tau(\bm t)$ be a tau-function of the BKP hierarchy satisfying $\tau(0)=1$,
then its BKP-affine coordinates $\{a_{n,m}\}_{n,m\geq 0}$
are the coefficients of the $Q_{(n)}(\bm t)$ and $Q_{(n,m)}(\bm t)$
in the Schur $Q$-expansion.
However,
in general it is hard to write down the explicit formula for the Schur $Q$-expansion
since one needs to compute of characters of projective representations.
In this subsection,
we will discuss another strategy to compute the BKP-affine coordinates --
we show that $\{a_{n,m}\}$ can be recovered from the wave function associated to $\tau(\bm t)$.
This method will lead to the notion of Kac-Schwarz operator of type $B$
(see the next subsection).
In what follows,
we will always consider BKP tau-functions $\tau(\bm t)$ satisfying the condition $\tau(0)=1$.

The BKP-wave function associated to a tau-function $\tau (\bm t)$ of the BKP hierarchy
is defined by Sato's formula \cite{djkm}:
\be
w_B (\bm t;z) = X_B(\bm t;z)\tau(\bm t) / \tau(\bm t),
\ee
where $X_B$ is the vertex operator:
\be
X_B(\bm t;z) = e^{\xi(\bm t;z)} e^{-\xi(\tpd,z^{-1})},
\ee
and
\begin{equation*}
\xi(\bm t,z) = \sum_{k>0:\text{ odd}} t_k z^{k}
= t_1 z+ t_3 z^3 +t_5 z^5+ \cdots,
\end{equation*}
and $\tpd = (2\pd_{t_1}, \frac{2}{3} \pd_{t_3},\frac{2}{5}\pd_{t_5},\cdots)$.
Or more explicitly,
\be
\label{eq-SatoFormula}
w_B (\bm t;z) = \exp\Big(\sum_{k=0}^\infty t_{2k+1} z^{2k+1}\Big)\cdot
\frac{\tau (t_1 - \frac{2}{z}, t_3 -\frac{2}{3z^3}, t_5-\frac{2}{5z^5},\cdots)}{\tau (t_1,t_3,t_5,\cdots)}.
\ee
Recall that the vertex operator $X_B$ is in fact the fermionic field $\phi(z)$ under
the boson-fermion correspondence of type $B$
(see \S \ref{sec-pre-bf}):
\begin{equation*}
\sigma_B(\phi(z) |U\rangle) = \frac{1}{\sqrt{2}} w \cdot X_B(\bm t,z) \sigma_B(|U\rangle),
\end{equation*}
where $w$ is a formal variable with $w^2=1$.
Then:
\be
\label{eq-wave-fermionic}
\begin{split}
w_B (\bm t;z) =& \sqrt{2} w^{-1} \cdot \sigma_B(\phi(z)e^A |0\rangle) /\tau(\bm t)\\
=&
\sqrt{2}  \cdot \langle 1 | \Gamma_+^B(\bm t) \phi(z)e^A |0\rangle /\tau(\bm t)\\
=& 2 \langle 0| \phi_0 \Gamma_+^B(\bm t) \phi(z)e^A |0\rangle /\tau(\bm t),
\end{split}
\ee
where
\begin{equation*}
A= \sum_{n,m\geq 0} a_{n,m} \phi_m\phi_n,
\end{equation*}
and $\{a_{n,m}\}_{n,m}$ are the BKP-affine coordinates of $\tau(\bm t)$,
satisfying $a_{n,m} = -a_{m,n}$.

Similar to the case of KP hierarchy (see e.g. \cite{av, zhou1}),
the BKP-affine coordinates can be recovered from the wave function in the following way:
\begin{Theorem}
\label{thm-wave&affine}
We have:
\be
\label{eq-thmwaveaffine}
\text{span} \big\{\pd_{t_1}^kw_B(0;z) \big\}_{k\geq 0} =
\text{span} \big\{ z^k + a_{k,0}+\sum_{i=1}^\infty 2(-1)^{i} a_{k,i} z^{-i} \big\}_{k\geq 0}.
\ee
as sub-vector spaces of $\cH = \bC[z]\oplus z^{-1} \bC [[z^{-1}]]$.
Or equivalently,
\be
\text{span} \big\{\pd_{t_1}^kw_B(0;z) \big\}_{k\geq 0} =
\text{span} \big\{ z^k + \sum_{i=1}^\infty 2(-1)^{i} (a_{k,i} - a_{k,0}a_{0,i})
z^{-i} \big\}_{k\geq 0}.
\ee
\end{Theorem}
\begin{Remark}
Knowing either one of the above two bases is actually equivalent to knowing all affine coordinates $\{a_{n,m}\}_{n,m\geq 0}$,
since one has $a_{n,m} = -a_{m,n}$.
\end{Remark}
\begin{proof}
Notice that by expanding the exponential $e^A$ directly and using the identity \eqref{eq-basic-VEV},
one can easily check:
\begin{equation*}
\langle \phi_{0} \phi(z) e^A \rangle =
\half + \sum_{i=1}^\infty (-1)^{i} a_{0,i} z^{-i},
\end{equation*}
and for every $k\geq 1$,
\begin{equation*}
\langle \phi_{-k} \phi(z) e^A \rangle = (-1)^k\cdot
\Big( z^k + a_{k,0}+\sum_{i=1}^\infty 2(-1)^{i} a_{k,i} z^{-i} \Big).
\end{equation*}
Thus in order to prove \eqref{eq-thmwaveaffine},
we only need to show:
\be
\label{eq-thmwaveaffine-pf}
\text{span}\{\pd_{t_1}^kw_B(0;z)\}_{k\geq 0} =
\text{span} \{\langle \phi_{-k} \phi(z) e^A \rangle \}_{k\geq 0}.
\ee
Now let us prove \eqref{eq-thmwaveaffine-pf}.
By \eqref{eq-wave-fermionic} we have:
\begin{equation*}
\begin{split}
\pd_{t_1}^k w_B( 0 ; z)
=& \pd_{t_1}^k \Big(
\frac{2 \langle \phi_0 e^{t_1 H_1} \phi(z) e^A \rangle}{ \tau(\bm t)} \Big)
\Big|_{\bm t=0} \\
=& \sum_{i=0}^k 2\binom{k}{i}
\pd_{t_1}^{k-i} \big(
\langle \phi_0 e^{t_1 H_1} \phi(z) e^A \rangle \big) \big|_{\bm t=0}
\cdot \pd_{t_1}^i \big( \frac{1}{\tau(\bm t)} \big) \big|_{\bm t=0} \\
=& \sum_{i=0}^k 2\binom{k}{i}
\cdot \langle \phi_0 H_1^{k-i} \phi(z) e^A \rangle
\cdot \pd_{t_1}^i \big( \frac{1}{\tau(\bm t)} \big) \big|_{\bm t=0}\\
=& 2 \langle \phi_0 H_1^k \phi(z) e^A \rangle
+ \sum_{i=1}^k 2\binom{k}{i}
\langle \phi_0 H_1^{k-i} \phi(z) e^A \rangle
\cdot \pd_{t_1}^i \big( \frac{1}{\tau(\bm t)} \big) \big|_{\bm t=0},
\end{split}
\end{equation*}
where $H_1$ is the bosonic operator defined by \eqref{eq-def-boson}.
Then we see that:
\begin{equation*}
\text{span} \big\{\pd_{t_1}^kw_B(0;z)\big\}_{k\geq 0}
=\text{span} \big\{\langle \phi_0 H_1^k \phi(z) e^A \rangle \big\}_{k\geq 0},
\end{equation*}
and thus it is sufficient to prove:
\be
\label{eq-thmwaveaffine-pf-2}
\text{span} \big\{\langle \phi_0 H_1^k \phi(z) e^A \rangle \big\}_{k\geq 0}
= \text{span} \big\{\langle \phi_{-k} \phi(z) e^A \rangle \big\}_{k\geq 0}.
\ee
Notice that:
\begin{equation*}
\begin{split}
H_{1} =&
\sum_{i=0}^\infty (-1)^{i+1} \phi_i \phi_{-i-1}\\
= &- \phi_0\phi_{-1} +\phi_1\phi_{-2} -\phi_2\phi_{-3} +\phi_3\phi_{-4}-\cdots,
\end{split}
\end{equation*}
and $\langle 0|\phi_n =0$ for $n>0$, therefore we can compute:
\begin{equation*}
\begin{split}
 \langle 0| \phi_0 H_1 =& -\langle 0| \phi_0^2 \phi_{-1} = -\half \langle 0| \phi_{-1},\\
 \langle 0 | \phi_0 H_1^2 =&
\sum_{i,j\geq 0} (-1)^{i+j} \langle 0 | \phi_0\phi_i\phi_{-i-1}\phi_j\phi_{-j-1}
= \half \langle 0| \phi_{-2},\\
 \langle 0 | \phi_0 H_1^3 = &
\sum_{i,j,k\geq 0} (-1)^{i+j+k} \langle 0 | \phi_0\phi_i\phi_{-i-1}\phi_j\phi_{-j-1}\phi_k\phi_{-k-1}\\
=& \half \langle 0 | \phi_{-2}\phi_0\phi_{-1} + \half \langle 0|  \phi_{-3}.\\
\end{split}
\end{equation*}
Inductively,
for every $k \geq 1$ one has:
\begin{equation*}
\langle 0| \phi_0 H_1^k =
\pm \half \langle 0| \phi_{-k} +\sum_{j\geq 2}
\sum_{\substack{n_1>\cdots>n_j\geq 0\\ n_1+\cdots +n_j =k}}
c_{n_1,\cdots,n_j} \langle 0| \phi_{-n_1}\cdots \phi_{n_j}
\end{equation*}
for some constants $c_{n_1,\cdots,n_j}$.
Thus to prove \eqref{eq-thmwaveaffine-pf-2},
it suffices to show that:
\be
\label{eq-thmwaveaffine-pf-3}
\langle \phi_{-n_1}\cdots \phi_{-n_j} \phi(z) e^A \rangle \in
\text{span} \big\{ \langle \phi_{-m} \phi(z) e^A \rangle \big\}_{m=0}^{k},
\ee
for every sequence $n_1>\cdots>n_j\geq 0$ with $n_1+\cdots +n_j =k$.
Now let us prove \eqref{eq-thmwaveaffine-pf-3}.
Notice that $\langle 0| = \langle 0| e^A$,
and thus:
\be
\label{eq-thmwaveaffine-pf-4}
\langle \phi_{-n_1}\cdots \phi_{-n_j} \phi(z) e^A \rangle
= \langle (e^{-A}\phi_{-n_1}e^A)
\cdots (e^{-A}\phi_{-n_j} e^A) (e^{-A} \phi(z) e^A) \rangle.
\ee
By \eqref{eq-lem-conj} we know that $e^{-A}\phi_{n} e^A$ is a linear summation of
the fermions $\{\phi_i\}_{i\in \bZ}$ for every $n\in \bZ$,
and thus one can apply Wick's Theorem to compute the right-hand side of \eqref{eq-thmwaveaffine-pf-4}.
Then one sees that \eqref{eq-thmwaveaffine-pf-4} equals to a linear combination of:
\begin{equation*}
\langle (e^{-A}\phi_{-n_i}e^A) (e^{-A}\phi(z)e^A) \rangle
=\langle \phi_{-n_i} \phi(z)e^A \rangle ,
\qquad i=1,2,\cdots,j,
\end{equation*}
where the coefficients are certain products of the numbers
\begin{equation*}
\langle (e^{-A}\phi_{-n_p}e^A) (e^{-A}\phi_{-n_q}e^A) \rangle
=\langle \phi_{-n_p}\phi_{-n_q}e^A \rangle.
\end{equation*}
Then \eqref{eq-thmwaveaffine-pf-3} is proved by induction on $k$.
\end{proof}

\subsection{Kac-Schwarz operators of type $B$}

Now using the above conclusion,
we can formulate the notion of Kac-Schwarz operators and quantum spectral curves of type $B$
for a BKP tau-function $\tau(\bm t)$ with $\tau(0)=1$ in terms of its
BKP-wave function and affine coordinates.

Let $\tau(\bm t)$ be such a tau-function,
and let $\{a_{n,m}\}_{n,m\geq 0}$ be its BKP-affine coordinates.
In what follows,
we will denote:
\be
\tilde a_{n,m} = a_{n,m} -a_{n,0}a_{0,m},
\qquad \forall n,m\geq 0,
\ee
and call $\{\tilde a_{n,m}\}_{n,m\geq 0}$ the modified affine coordinates of $\tau(\bm t)$.
Then the two sets of coordinates $\{a_{n,m}\}_{n,m\geq 0}$ and $\{\tilde a_{n,m}\}_{n,m\geq 0}$
are uniquely determined by each other.
The modified affine coordinates satisfy the following constraints:
\be
\label{eq-mofif-condition}
\begin{split}
&\tilde a_{0,n} = -\tilde a_{n,0};\\
&\tilde a_{n,m} + \tilde a_{m,n} = -2\tilde a_{n,0}\tilde a_{0,m}.
\end{split}
\ee
The discussions in the previous subsection tell us that
we can canonically associate a linear subspace
\begin{equation*}
U_\tau \subset \cH = \bC[z]\oplus z^{-1}\bC[[z^{-1}]]
\end{equation*}
to $\tau(\bm t)$ in the following way
(see Theorem \ref{thm-wave&affine}):
\be
U_\tau =
\text{span} \big\{\pd_{t_1}^kw_B(0;z) \big\}_{k\geq 0} =
\text{span} \big\{ \Phi_k^B(z)  \big\}_{k\geq 0}.
\ee
where:
\be
\Phi_k^B(z) =
z^k + \sum_{i=1}^\infty 2(-1)^{i} \tilde a_{k,i}
z^{-i},
\qquad k\geq 0.
\ee
\begin{Remark}
The subspace $ U_\tau \subset \cH$ is actually a point in the big cell $Gr_{(0)}$
of the usual Sato Grassmannian
since it is clear that $\pi_+ | _{U_\tau}$ is an isomorphism,
where $\pi_+$ is the natural projection
\begin{equation*}
\pi_+ : \cH \to \cH_+ = \bC[z].
\end{equation*}
The KP-affine coordinates (see \cite{zhou1} for notations) of this point is given by:
\begin{equation*}
a_{n,m}^{KP} = 2(-1)^{m+1} \cdot \tilde a_{n,m+1}.
\end{equation*}
This defines an embedding of $Gr_{\cH_\phi^0}^0 (\cH_\phi)$ into $Gr_{(0)}$.
\end{Remark}

In what follows,
we will consider the notion of Kac-Schwarz operators for the isotropic Sato Grassmannian
in the following sense.
We say that an operator $P$ (acting on formal Laurent series in $z$) is a Kac-Schwarz operator of type $B$
for the tau-function $\tau(\bm t)$,
if it satisfies:
\be
P(U_\tau) \subset U_\tau.
\ee
Similar to the case of KP hierarchy and the usual Sato Grassmannian
(see \cite{zhou6, al1, al2} for some examples),
here we are especially interested in those BKP tau-functions for which
there are two Kac-Schwarz operators $(P,Q)$ of type $B$,
such that:
\be
\label{eq-KS-conditionPQ}
\begin{split}
& P(\Phi_0^B) = 0;\\
& Q(\Phi_k^B) - c_{k+1}\cdot \Phi_{k+1}^B \in \text{span}\{\Phi_0^B,\Phi_1^B,\cdots,\Phi_{k}^B\}.
\end{split}
\ee
where $\{c_{k+1}\}$ are some nonzero constants.
Moreover,
we hope that they satisfy the canonical commutation relation:
\be
[P,Q]=1.
\ee
If the equation $P(\Phi_0^B)=0$ has a unique solution
$\Phi_0^B \in 1+ z^{-1}\bC[[z^{-1}]]$,
then $\{\tilde a_{n,m}\}$ and hence the tau-function $\tau(\bm t)$ are uniquely determined,
since $\Phi_k^B(z)$ for every $k>0$ can be recursively computed by applying $Q$
to $\Phi_{k-1}^B(z)$.

\subsection{Quantum spectral curves}

Similar to the case of KP tau-functions,
we call an operator $\Delta$ (acting on formal Laurent series in $z$)
a quantum spectral curve for a BKP tau-function $\tau$,
if it annihilates the first basis vector for $U_\tau$:
\be
\label{eq-def-QSC-B}
\Delta (\Phi_0^B) = 0.
\ee
In particular,
if there exists a pair of Kac-Schwarz operators $(P,Q)$ satisfying \eqref{eq-KS-conditionPQ},
then $P$  is a quantum spectral curve of type $B$.
In the rest of this subsection,
we explain that this definition of quantum spectral curves actually match with
the construction of Gukov-Su{\l}kowski's \cite{gs}.

First we rewrite the definition \eqref{eq-def-QSC-B} in terms of the principal specialization of $\tau(\bm t)$.
Notice that the first basis vector $\Phi_0^B(z)$ is given by:
\begin{equation*}
\Phi_0^B(z) = 1+\sum_{i=0}^\infty 2(-1)^i a_{0,i} z^{-i}
= 2\langle \phi_0 \phi(z) e^A \rangle =
w_B(0;z),
\end{equation*}
then \eqref{eq-def-QSC-B} is equivalent to:
\be
\Delta (w_B(0;z)) =0.
\ee
Moreover,
by Sato's formula \eqref{eq-SatoFormula} we know that $w_B(0;z)$ is given by
the principal specialization $t_k = -\frac{2}{kz^k}$ (for every odd $k>0$) of the tau-function $\tau(\bm t)$,
therefore the definition \eqref{eq-def-QSC-B} is equivalent to:
\be
\Delta \Big( \tau(-\frac{2}{z}, -\frac{2}{3z^3}, -\frac{2}{5z^5}, \cdots)\Big) =0.
\ee

Now let $\cC$ be a plane curve.
Similar to the case of KP tau-functions
(see \S \ref{sec-intro-QSC} for a brief review),
if a BKP tau-function $\tau(\bm t)$ is reconstructed from
the E-O topological recursion on $\cC$,
then the B-A function constructed from the E-O invariants using Gukov-Su{\l}kowski's ansatz
(with a possible modification of degrees of $z$ in \eqref{eq-def-Sn},
since the BKP-time variables are indexed by odd numbers)
can be regarded as the principal specialization of $\tau(\bm t)$,
which coincides with the first basis vector in $U_\tau$.
In this case,
one may expect to find a quantum spectral curve $\Delta(\Phi_0^B) = 0$
such that the semi-classical limit of the operator $\Delta$
recovers the classical spectral curve $\cC$,
according to Gukov-Su{\l}kowski's conjecture.

In the next two sections,
we will consider the example of spin single Hurwitz numbers with completed cycles.
We construct a pair of Kac-Schwarz operators $(P,Q)$ of type $B$ for the corresponding BKP tau-function,
such that the operator $P$ also serves as the quantum spectral curve.

\section{BKP-Affine Coordinates for Spin Single Hurwitz Numbers}
\label{sec-sH-affine}

It is known that the generating series of spin single Hurwitz numbers with completed cycles
is a tau-function of the BKP hierarchy \cite{le, gkl}.
In this section,
we first review some facts about spin Hurwitz numbers,
and then compute the BKP-affine coordinates of the associated tau-function.

\subsection{Spin Hurwitz numbers with completed cycles}

In this subsection,
we recall the notion of spin Hurwitz numbers and
the fermionic representation of their generating series.

The ordinary Hurwitz numbers,
which count the numbers of branched covers between Riemann surfaces with specified ramification types,
were first introduced by Hurwitz in \cite{hur}.
In \cite{eop},
Eskin-Okounkov-Pandharipande introduced a spin structure (or theta characteristic) on the source,
and introduced a new type of Hurwitz numbers in this way.
These new Hurwitz numbers are called spin Hurwitz numbers.
Similar to the case of ordinary Hurwitz numbers \cite{pa, Ok1, ssz},
the generating series of disconnected spin Hurwitz numbers
are also controlled by some integrable hierarchies,
see e.g. \cite{le, gkl,  mmn, mmno}.
Denote by $h_{g;\mu}^{\bullet, r,\vartheta}$  the disconnected spin single Hurwitz number
with $(r+1)$-completed cycles (where $r$ is even),
of genus $g$ and ramification type $\mu$
(see \cite[\S 4]{gkl} for the detailed definition and computations via Gunningham's formula \cite{gu}).
And let
\be
\tau^{(r,\vartheta)}(\bm t) = \sum_{g;\mu}
2^{g-1} \beta^b   h_{g;\mu}^{\bullet, r,\vartheta} \cdot
 \frac{p^{|\mu|}}{|\mu|!} \cdot
\frac{ p_{\mu} }{ l(\mu)! }
\ee
be the generating series of all such disconnected spin single Hurwitz numbers,
where $p$ and $(p_1,p_2,\cdots)$ are some formal variable,
and
$p_\mu = p_{\mu_1} p_{\mu_2} \cdots p_{\mu_l}$,
where $p_n = n\cdot t_n$.
The summation in the right-hand side is over all odd partitions $\mu$
and all integers $g$ (genus).
The number $b$ of $(r+1)$-completed cycle ramification points
is determined by $g$ and $\mu$ via the Riemann-Hurwitz formula:
\be
\label{eq-R-H}
b= (2g-2 +l(\mu) +|\mu|) /r.
\ee
The generating series $\tau^{(r,\vartheta)}(\bm t)$ is known to be a tau-function of
the BKP hierarchy,
which can be represented as a fermionic vacuum expectation value as follows
(see \cite{le} and \cite[Theorem 6.20]{gkl}):
\be
\label{eq-tau-rtheta-1}
\tau^{(r,\vartheta)} =
\langle \Gamma_+ (\bm t)
 \exp\big(\beta\frac{\hF_{r+1}}{r+1}\big)
e^{pH_{-1}} \rangle,
\ee
where $\hF_{r+1}$ is the following operator on the fermionic Fock space $\cF_B$:
\be
\label{eq-tau-rtheta-2}
\hF_{r+1}
= \sum_{k>0} (-1)^k k^{r+1} :\phi_k\phi_{-k}:.
\ee

\subsection{Computation of BKP-affine coordinates}

In this subsection,
we compute the affine coordinates of the tau-function $\tau^{(r,\vartheta)} (\bm t)$.

Let $f:\bZ_{>0} \to \bC$ be a function defined on the set of positive integers.
Consider a tau-function of the BKP hierarchy of the following form:
\be
\label{eq-tau-diagonal}
\tau (\bm t) = \langle 0 | \Gamma_+^B(\bm t) e^\hf e^{p H_{-1}} |0\rangle,
\ee
where $p$ is a formal variable,
and
\be
\hf = \sum_{m=1}^\infty (-1)^m f(m) :\phi_m\phi_{-m}:
= \sum_{m=1}^\infty (-1)^m f(m) \phi_m \phi_{-m}.
\ee
Then we have:
\begin{Proposition}
The BKP-affine coordinates of $\tau(\bm t)$ are given by:
\be
\begin{split}
&a_{0,n} =  \frac{p^n}{ 2\cdot n!}\cdot e^{f(n)},\qquad n>0;\\
&a_{n,m} = \frac{p^{m+n}}{4\cdot m!\cdot n!} \cdot \frac{m-n}{m+n} \cdot e^{f(m)+f(n)},
\qquad m>n>0,
\end{split}
\ee
and $a_{n,m} = -a_{m,n}$ for every $n,m\geq 0$.
\end{Proposition}
\begin{proof}
Recall that the BKP-affine coordinates of a tau-function are the coefficients of the Schur $Q$-functions $Q_\mu$
with $l(\mu ) \leq 2$ in its Schur $Q$-function expansion
(see \cite{wy, hb} for details; and see \S \ref{sec-pre-affine} for a brief review).
That is equivalent to say,
we need to expand the corresponding vector in fermionic Fock space as a summation of basis vectors
$|\mu\rangle$ and take the coefficients of $|\mu\rangle$ with $l(\mu ) \leq 2$.

Let $\bm{\tilde t} = (\tilde t_1,\tilde t_3,\tilde t_5,\cdots)$ be a sequence of formal variables.
Notice that one has (see e.g. \cite[\S 3]{le}):
\begin{equation*}
\Gamma_-^B (\bm{\tilde t}) |0\rangle
= \sum_{\mu\in DP} 2^{-\half l(\mu)}
\cdot Q_\mu (\half \bm{\tilde t}) |\mu\rangle.
\end{equation*}
Then by evaluating at $\bm{\tilde t} = (p,0,0,0,\cdots)$, one gets:
\begin{equation*}
e^{p H_{-1}} |0\rangle = \sum_{\mu\in DP}
2^{-\half l(\lambda)} \cdot
Q_\mu (\frac{\delta_{k,1}}{2} p) |\mu\rangle.
\end{equation*}
The following identity is known in literatures
(see \cite[(56)]{mm}):
\begin{equation*}
Q_\mu (\frac{\delta_{k,1}}{2} p) =
\frac{p^{|\mu|}}{ \prod_{i=1}^{l(\mu)} \mu_i!}
\cdot \prod_{i<j} \frac{\mu_i-\mu_j}{\mu_i+\mu_j},
\end{equation*}
hence we have:
\begin{equation*}
e^{p H_{-1}} |0\rangle
= \sum_{\mu\in DP}
2^{-\half l(\mu)} \cdot
\frac{p^{|\mu|}}{ \prod_{i=1}^{l(\mu)} \mu_i!}
\cdot \prod_{i<j} \frac{\mu_i-\mu_j}{\mu_i+\mu_j}
 |\mu\rangle,
\end{equation*}
and then:
\be
\label{eq-pf-hf-affine}
e^\hf e^{pH_{-1}} |0\rangle
= \sum_{\mu\in DP}
2^{-\half l(\mu)} \cdot
\frac{p^{|\mu|}}{ \prod_{i=1}^{l(\mu)} \mu_i!}
\cdot \prod_{i<j} \frac{\mu_i-\mu_j}{\mu_i+\mu_j} \cdot
e^\hf |\mu\rangle.
\ee

Now recall that the vector $|\mu\rangle \in \cF_B$ is defined by:
\ben
|\mu\rangle = \begin{cases}
\phi_{\mu_1} \phi_{\mu_2} \cdots \phi_{\mu_n} |0\rangle, & \text{ $n$ even;}\\
\sqrt{2} \phi_{\mu_1} \phi_{\mu_2} \cdots \phi_{\mu_n} \phi_0 |0\rangle, & \text{ $n$ odd.}
\end{cases}
\een
Using the Baker-Campbell-Hausdorff formula,
one has (see e.g. \cite[Lemma 3.1]{wy3}):
\begin{equation*}
\begin{split}
&e^{\hf} \phi_k e^{-\hf} = e^{f(k)} \phi_k,
\qquad \forall k>0;\\
&e^{\hf} \phi_0 e^{-\hf} = \phi_0.
\end{split}
\end{equation*}
Thus we have
\begin{equation*}
\begin{split}
e^\hf |\mu\rangle =&
(e^\hf\phi_{\mu_1} e^{-\hf}) (e^\hf\phi_{\mu_2}e^{-\hf})
 \cdots (e^\hf \phi_{\mu_n}e^{-\hf}) |0\rangle\\
 =& e^{f(\mu_1) + f(\mu_2) +\cdots +f(\mu_k)} |\mu\rangle
 \end{split}
\end{equation*}
for $n$ even,
and
\begin{equation*}
\begin{split}
e^\hf |\mu\rangle =& \sqrt{2} \cdot
(e^\hf\phi_{\mu_1} e^{-\hf}) (e^\hf\phi_{\mu_2}e^{-\hf})
 \cdots (e^\hf \phi_{\mu_n}e^{-\hf}) (e^\hf \phi_{0}e^{-\hf}) |0\rangle\\
 =& e^{f(\mu_1) + f(\mu_2) +\cdots +f(\mu_k)} |\mu\rangle
 \end{split}
\end{equation*}
for $n$ odd.
Plugging this into \eqref{eq-pf-hf-affine},
we get:
\begin{equation*}
e^\hf e^{pH_{-1}} |0\rangle
= \sum_{\mu\in DP}
2^{-\half l(\mu)} \cdot
\frac{p^{|\mu|}}{ \prod_{i=1}^{l(\mu)} \mu_i!}
\cdot \prod_{i<j} \frac{\mu_i-\mu_j}{\mu_i+\mu_j} \cdot
e^{\sum_{i=1}^{l(\mu)} f(\mu_i)}
|\mu\rangle.
\end{equation*}
Notice that the BKP-affine coordinates are the coefficients of $|\mu\rangle$
with $l(\mu) \leq 2$
(up to a factor $2^{\half l(\lambda)}$, see \S \ref{sec-pre}),
and thus the conclusion is proved.
\end{proof}

Now by taking $f:\bZ_{>0} \to \bC$ to be
(see \eqref{eq-tau-rtheta-1}, \eqref{eq-tau-rtheta-2}):
\be
\label{eq-def-f-rtheta}
f^{(r,\vartheta)}:\bZ_{>0} \to \bC,
\qquad
m \mapsto \beta \cdot \frac{m^{r+1}}{r+1},
\ee
we obtain the following:
\begin{Corollary}
The BKP-affine coordinates of $\tau^{(r,\vartheta)}$ are
$a_{0,0}^{(r,\vartheta)} = 0$,
and:
\be
\label{eq-affinecor-sH}
\begin{split}
&a_{0,n}^{(r,\vartheta)} = -a_{n,0}^{(r,\vartheta)} =
  \frac{p^n}{ 2\cdot n!}\cdot \exp\big( \beta\frac{n^{r+1}}{r+1} \big),
\qquad \forall n>0;\\
&a_{n,m}^{(r,\vartheta)} = \frac{p^{m+n}}{4\cdot m!\cdot n!} \cdot \frac{m-n}{m+n} \cdot
\exp \big( \beta\frac{m^{r+1} + n^{r+1}}{r+1} \big),
\qquad \forall m,n>0.
\end{split}
\ee
\end{Corollary}

\subsection{A formula for connected spin Hurwitz numbers}

As a byproduct of the result in last subsection,
we obtain an explicit formula for connected spin Hurwitz numbers.
In a previous work \cite{wy},
the second and third authors have derived a formula for the connected bosonic $n$-point functions
of a BKP tau-function in terms of its BKP-affine coordinates,
see \cite[Theorem 1.1]{wy}:
\begin{Theorem}
[\cite{wy}]
Let $A(w,z),\wA(w,z)$ be the following generating series of the BKP-affine coordinates $\{a_{n,m}\}$
of a BKP tau-function $\tau(\bm t)$ satisfying $\tau(0)=1$:
\begin{equation*}
\begin{split}
&A(w,z) = \sum_{n,m>0} (-1)^{m+n+1} \cdot a_{n,m} w^{-n} z^{-m}
-\half \sum_{n>0}(-1)^n \cdot a_{n,0} (w^{-n}-z^{-n}),\\
&\wA(w,z) =  A(w,z) -\frac{1}{4} -\half\sum_{i=1}^\infty (-1)^{i} w^{-i} z^i.
\end{split}
\end{equation*}
Then:
\begin{equation*}
\sum_{i> 0: \text{ odd}}
\frac{\pd \log\tau(\bm t)}{\pd t_{i}} \bigg|_{\bm t=0}
\cdot z^{-i}
=A(-z,z),
\end{equation*}
and for $n\geq 2$,
\be
\label{eq-npt-general}
\begin{split}
&\sum_{i_1,\cdots,i_n> 0: \text{ odd}}
\frac{\pd^n \log\tau(\bm t)}{\pd t_{i_1}\cdots \pd t_{i_n}} \bigg|_{\bm t=0}
\cdot z_1^{-i_1}\cdots z_n^{-i_n}
=
-\delta_{n,2} \cdot i_{z_1,z_2}
\frac{z_1z_2(z_2^2+z_1^2)}{2(z_1^2-z_2^2)^2}\\
&\qquad\qquad\qquad
+ \sum_{\substack{ \sigma: \text{ $n$-cycle} \\ \epsilon_2,\cdots,\epsilon_n \in\{\pm 1\}}}
(-\epsilon_2\cdots\epsilon_n) \cdot
\prod_{i=1}^n \xi(\epsilon_{\sigma(i)} z_{\sigma(i)}, -\epsilon_{\sigma(i+1)} z_{\sigma(i+1)}),
\end{split}
\ee
where
\begin{equation*}
i_{z_1,z_2}
\frac{z_1z_2(z_2^2+z_1^2)}{2(z_1^2-z_2^2)^2} =
\sum_{n>0:\text{ odd}} \frac{n}{2} z_1^{-n} z_2^n,
\end{equation*}
and $\xi$ is given by:
\begin{equation*}
\begin{split}
\xi(\epsilon_{\sigma(i)} z_{\sigma(i)}, -\epsilon_{\sigma(i+1)} z_{\sigma(i+1)}) =
\begin{cases}
\wA (\epsilon_{\sigma(i)} z_{\sigma(i)}, -\epsilon_{\sigma(i+1)} z_{\sigma(i+1)}),
& \sigma(i)<\sigma(i+1);\\
-\wA( -\epsilon_{\sigma(i+1)} z_{\sigma(i+1)} ,\epsilon_{\sigma(i)} z_{\sigma(i)}),
& \sigma(i)>\sigma(i+1),
\end{cases}
\end{split}
\end{equation*}
and we use the conventions
$\epsilon_{1} :=1$ and
$\sigma(n+1):=\sigma(1)$.
\end{Theorem}

The above formula for $n\geq 2$ can be further simplified.
In fact,
for a fixed cycle $\sigma$ and a fixed $j$,
the formal variable $z_j$ appears only in two terms
\begin{equation*}
\xi (\pm z_i, -\epsilon_j z_j) \cdot \xi (\epsilon_j z_j, \pm z_k)
\end{equation*}
in $\prod_{i=1}^{n} \xi(\epsilon_{\sigma(i)} z_{\sigma(i)}, -\epsilon_{\sigma(i+1)} z_{\sigma(i+1)})$
(where $i$ and $k$ are adjacent to $j$ in this cycle $\sigma$),
hence replacing $\epsilon_j$ by $-\epsilon_j$ is equivalent to replacing $z_j$ by $-z_j$.
Then replacing $\epsilon_j$ by $-\epsilon_j$ does not change the terms
with odd orders in $z_j$ in the product
\begin{equation*}
\epsilon_2\cdots\epsilon_{n+m}
\prod_{i=1}^{n+m} \xi(\epsilon_{\sigma(i)} z_{\sigma(i)}, -\epsilon_{\sigma(i+1)} z_{\sigma(i+1)}).
\end{equation*}
Moreover,
we already know that the order of $z_j$ in the left-hand side of
\eqref{eq-npt-general} is always an odd number,
and therefore we can simply take $\epsilon_2 = \cdots = \epsilon_{n} =1$
in the right-hand side and then restrict to terms of odd degrees.
In this way we obtain:
\be
\label{eq-npt-general-2}
\begin{split}
&\sum_{i_1,\cdots,i_n> 0: \text{ odd}}
\frac{\pd^n \log\tau(\bm t)}{\pd t_{i_1}\cdots \pd t_{i_n}} \bigg|_{\bm t=0}
\cdot z_1^{-i_1}\cdots z_n^{-i_n}
=
-\delta_{n,2} \cdot
\frac{z_1z_2(z_2^2+z_1^2)}{2(z_1^2-z_2^2)^2}\\
&\qquad\qquad\qquad\qquad\qquad\qquad\qquad
- 2^{n-1}  \cdot \Big[ \sum_{ \sigma: \text{ $n$-cycle} }
\prod_{i=1}^n \xi( z_{\sigma(i)}, - z_{\sigma(i+1)})
\Big]_{\text{odd}},
\end{split}
\ee
where $[\cdot]_{\text{odd}}$ means taking the terms
of odd degrees in every $z_i$.

Now consider the tau-function $\tau^{(r,\vartheta)}$.
In this case the generating series $A(w,z)$ and $\wA(w,z)$ are
(for simplicity here we take $p=1$):
\begin{equation*}
\begin{split}
&A^{(r,\vartheta)}(w,z) = \sum_{(m,n)\in \bZ_{\geq 0}^2 \backslash \{(0,0)\}}
(-1)^{m+n+1}\cdot \frac{e^{f^{(r,\vartheta)}(m)+f^{(r,\vartheta)}(n)}}{4 \cdot m!\cdot n!} \cdot \frac{m-n}{m+n}
w^{-n} z^{-m},\\
&\wA^{(r,\vartheta)}(w,z) =  A^{(r,\vartheta)}(w,z) -\frac{1}{4} -\half\sum_{i=1}^\infty (-1)^{i} w^{-i} z^i,
\end{split}
\end{equation*}
and the above theorem gives:
\begin{Proposition}
Denote by
\begin{equation*}
H_\mu^\circ(\beta) = \sum_{g} 2^{g-1}\cdot h_{g,\mu}^{\circ,r,\vartheta} \cdot  \frac{ \beta^b}{|\mu|! \cdot l(\mu)!}
\end{equation*}
the generating series of connected spin single Hurwitz numbers
with ramification type $\mu$,
where $b$ and $g$ are related by \eqref{eq-R-H}.
Then we have:
\be
H_{(n)}^\circ(\beta)
= \frac{1}{n} \sum_{k=0}^n
(-1)^{k+1} \frac{(2k-n) e^{f^{(r,\vartheta)}(k)+f^{(r,\vartheta)}(n-k)}}{4n\cdot k!\cdot (n-k)!},
\qquad \forall n>0\text{ odd};
\ee
and for an odd partition $\mu=(\mu_1,\cdots,\mu_n) = (1^{m_1} 3^{m_3} 5^{m_5}\cdots)$ with $n\geq 2$,
\be
\begin{split}
H^\circ_\mu(\beta)
=-\frac{2^{n-1}}{z_\mu} \cdot{Coeff}_{\prod_{i=1}^{n} z_i^{-\mu_i}}
\Big(\sum_{ \sigma: \text{ $n$-cycle}}
\prod_{i=1}^n \xi(z_{\sigma(i)}, -z_{\sigma(i+1)})\Big),
\end{split}
\ee
where Coeff means taking the coefficient,
and $z_\mu = \prod_{j\geq 1} m_j! \cdot j^{m_j}$.
Here $\xi$ is:
\begin{equation*}
\begin{split}
\xi( z_{\sigma(i)}, - z_{\sigma(i+1)}) =
\begin{cases}
\wA^{(r,\vartheta)} ( z_{\sigma(i)}, - z_{\sigma(i+1)}),
& \sigma(i)<\sigma(i+1);\\
-\wA^{(r,\vartheta)}( - z_{\sigma(i+1)} , z_{\sigma(i)}),
& \sigma(i)>\sigma(i+1),
\end{cases}
\end{split}
\end{equation*}
and we use the convention
$\sigma(n+1):=\sigma(1)$.
\end{Proposition}

\section{Kac-Schwarz Operators and Quantum Spectral Curve for Spin Single Hurwitz Numbers}
\label{sec-sH-KS}

In this section,
we use the explicit expressions for the BKP-affine coordinates of $\tau^{(r,\vartheta)}$
to find a pair of Kac-Schwarz operators of type $B$
satisfying the conditions \eqref{eq-KS-conditionPQ} and $[P,Q]=1$.
In particular,
$P$ gives the quantum spectral curve for spin single Hurwitz numbers.

\subsection{Main results}
\label{sec-sH-QSC-main}

Let $r>0$ be even,
and let $f^{(r,\vartheta)}$ be the polynomial function
\be
f^{(r,\vartheta)}(m) = \beta \frac{m^{r+1}}{r+1}, \qquad m\in \bZ.
\ee
Denote by $\{\Phi_k^{(r,\vartheta)}\}_{k\geq 0}$ the following basis vectors
for $U_{\tau^{(r,\vartheta)}}$:
\be
\begin{split}
\Phi_k^{(r,\vartheta)} =&
z^k + \sum_{i=1}^\infty 2(-1)^{i}
(a_{k,i}^{(r,\vartheta)} - a_{k,0}^{(r,\vartheta)}a_{0,i}^{(r,\vartheta)})
z^{-i}\\
=&
z^k+\sum_{i\ge1}2(-1)^i \cdot \frac{p^{k+i}}{4\cdot k! \cdot i!} \cdot
\frac{2i}{i+k}e^{f^{(r,\vartheta)}(k)+f^{(r,\vartheta)}(i)}z^{-i}.
\end{split}
\ee
Then we have:
\begin{Theorem}
\label{thm-sH-main}
Let $P,Q$ be the following operators:
\be
\label{eq-KS-spinH}
\begin{split}
P= &\exp\Big( \frac{\beta}{r+1} \cdot \sum_{i=0}^{r}z^{-1}\big(z\frac{\partial}{\partial z}\big)^iz\big(z\frac{\partial}{\partial z}\big)^{r-i}\Big)\partial_z\\
& \quad -p\exp\Big(\frac{2\beta}{r+1}\cdot \sum_{i=0}^{r}z^{-2}\big(z\frac{\partial}{\partial z}\big)^iz^{2}\big(z\frac{\partial}{\partial z}\big)^{r-i}\Big)z^{-2};\\
Q= &\exp\Big(- \frac{\beta}{r+1}\cdot \sum_{i=0}^{r}z\big(z\frac{\partial}{\partial z}\big)^iz^{-1}\big(z\frac{\partial}{\partial z}\big)^{r-i}\Big)z,
\end{split}
\ee
then $P$ and $Q$ are Kac-Schwarz operators of the BKP tau-function $\tau^{(r,\vartheta)}$
for spin single Hurwitz numbers with $(r+1)$-completed cycles.
Moreover,
we have:
\be
\label{eq-PQ-action}
\begin{split}
&P(\Phi_k^{(r,\vartheta)})=  k e^{f^{(r,\vartheta)}(k)-f^{(r,\vartheta)}(k-1)} \Phi_{k-1}^{(r,\vartheta)}
-p e^{f^{(r,\vartheta)}(k)-f^{(r,\vartheta)}(k-2)} \Phi_{k-2}^{(r,\vartheta)},\\
&Q(\Phi_k^{(r,\vartheta)}) =
e^{f^{(r,\vartheta)}(k)-f^{(r,\vartheta)}(k+1)}\Phi_{k+1}^{(r,\vartheta)}
-\frac{p^{k+1} \cdot e^{f^{(r,\vartheta)}(k)}}{(k+1)!}\Phi_0^{(r,\vartheta)},
\end{split}
\ee
for every $k\geq 0$,
where we use the conventions $\Phi_{-1}^{(r,\vartheta)} = \Phi_{-2}^{(r,\vartheta)} =0$.
In particular,
\be
P(\Phi_0^{(r,\vartheta)}) =0.
\ee
Furthermore,
they satisfy the canonical commutation relation:
\be
[P,Q]=1.
\ee
\end{Theorem}
The proof of this theorem will be given in the next subsection.

Notice that the equation $P(\Phi_0^{(r,\vartheta)}) =0$ is the quantum spectral curve
of the BKP tau-function $\tau^{(r,\vartheta)}$.
Now we look for the semi-classical limit of the operator $P$.
Let $(x,y)$ be the coordinates on the complex $2$-space $\bC^* \times \bC$ equipped with
the symplectic form $\omega = \frac{\sqrt{-1}}{\hbar}du\wedge dv$,
where $x=e^u$ and $y=v$.
Consider the canonical quantization:
\begin{equation*}
\hat x =x, \qquad\qquad \hat y = \hbar \frac{d}{du} = \hbar x\frac{d}{dx}.
\end{equation*}
Denote $x=z^{-1}$, and $\beta=\hbar^r$, $p=\frac{1}{\hbar}$.
After substituting these changes of variables into $P$ and multiplying by an additional $\hbar$,
we get:
\begin{equation*}
\hbar P
=-\exp\Big(\frac{\widehat{x}\sum_{i=0}^{r}\widehat{y}^i \widehat{x}^{-1}
\widehat{y}^{r-i}}{r+1}\Big)\widehat{x}\widehat{y}-
\exp\Big(2\frac{\widehat{x}^{2}\sum_{i=0}^{r}\widehat{y}^i
\widehat{x}^{-2} \widehat{y}^{r-i}}{r+1}\Big)\widehat{x}\widehat{x}.
\end{equation*}
Notice that:
\begin{equation*}
\exp({y^r})=\exp\Big(\frac{\sum_{i=0}^{r}y^iy^{r-i}}{r+1}\Big)
=\exp\Big(\frac{\sum_{i=0}^{r}x^{-l}y^ix^{l}y^{r-i}}{r+1}\Big),
\end{equation*}
and thus $\hbar P$ is the canonical quantization of the following classical potential:
\be
H(x,y)=-e^{y^r}xy-e^{y^{2r}}x^2.
\ee
The zero-locus $H(x,y)=0$ of this classical potential is a plane curve:
\be
\label{eq-classical-curve}
x=-ye^{-y^r}.
\ee
This is the semi-classical limit of the operator $P$.

Furthermore,
it has been conjectured by Giacchetto, Kramer, Lewa\'nski \cite{gkl}
and proved by Alexandrov and Shadrin \cite{as} that
the spin single Hurwitz numbers with $(r+1)$-completed cycles can be reconstructed from
the Eynard-Orantin topological recursion on this curve.
Thus we conclude that:
\begin{Corollary}
The equation $P(\Phi_0^{(r,\vartheta)}) = 0$ is the quantum spectral curve
of \eqref{eq-classical-curve} in the sense of Gukov-Su{\l}kowski \cite{gs}.
\end{Corollary}
Note that in this case no higher order quantum corrections are involved.

\begin{Remark}
The spectral curve in \cite{gkl, as} is
(see \cite[(7.1)]{gkl}):
\begin{equation*}
x (z) = \log (z) - z^r,
\qquad\qquad
y(z) = z.
\end{equation*}
This coincides with \eqref{eq-classical-curve} after the change of coordinates $x\mapsto -e^x,y\mapsto y$.
\end{Remark}

\begin{Remark}
The Bergman kernel for the E-O recursion in \cite{as, gkl} is:
\begin{equation*}
\cB (z_1,z_2) = \half \Big(
\frac{1}{(z_1-z_2)^2} + \frac{1}{(z_1+z_2)^2} \Big) dz_1dz_2.
\end{equation*}
which differs from the correction term at $n=2$ in the right-hand of \eqref{eq-npt-general}
only by a factor $\half z_1z_2$.
It is worth noting that in Zhou's formula for connected $n$-point functions
of KP tau-functions \cite{zhou1},
the correction term at $n=2$ is $\frac{1}{(z_1-z_2)^2}$,
which coincide with the original Bergman kernel in E-O recursion.
We may hope to gain a better understanding about
the connections between such formulas for connected $n$-point functions
and the E-O topological recursion in future.

It is known that in the case of Witten-Kontsevich tau-function \cite{wit, kon},
the above two types of Bergman kernel are equivalent since they give the same recursion kernel,
see Zhou \cite{zhou3} for details.
\end{Remark}

\subsection{Proof of Theorem \ref{thm-sH-main}}

Now we prove Theorem \ref{thm-sH-main}.
For simplicity,
in this subsection we will use the convention $f = f^{(r,\vartheta)}$,
and denote:
\begin{equation*}
E_l=\exp\Big( \frac{l\cdot \beta}{r+1}\cdot \sum_{i=0}^{r}z^{-l}
\big(z\frac{\partial}{\partial z}\big)^iz^{l}
\big(z\frac{\partial}{\partial z}\big)^{r-i}
\Big), \qquad l\in \mathbb{Z}.
\end{equation*}
Then it is clear that:
\be
\label{eq-PQ-E}
P=E_1\partial_z-pE_2z^{-2},
\qquad Q=E_{-1}z.
\ee
First we prove the following:
\begin{Lemma}
\label{lem-KS-action}
The action of $P$ and $Q$ on $z^k$ are:
\begin{equation*}
\begin{split}
P (z^k)
   =& ke^{f(k)-f(k-1)}z^{k-1}-pe^{f(k)-f(k-2)}z^{k-2},\\
Q (z^k)
   =&  e^{f(k)-f(k+1)}z^{k+1},
\end{split}
\end{equation*}
for every $k\in \bZ$.
\end{Lemma}
\begin{proof}
First observe that the action of the operator
$z^{-l}  (z\partial_z)  z^{l}$
on $z^{k}$ is:
\begin{equation*}
\big( z^{-l} (z\partial_z) z^{l} \big) ( z^k )
=(k+l) z^k =(z\partial_z+l)(z^k),
\end{equation*}
and thus $z^{-l}\circ (z\partial_z) \circ z^{l} = z\partial_z+l$, and then:
\begin{equation*}
\begin{split}
(z\partial_z+l)^{r+1}-(z\partial_z)^{r+1}=&
l \cdot \sum_{i=0}^r  (z\partial_z+l)^i (z\partial_z)^{r-i}\\
=& l\cdot \sum_{i=0}^{r}z^{-l}(z\partial_z)^iz^{l} (z\partial_z)^{r-i}.
\end{split}
\end{equation*}
Then the action of $E_l$ on $z^k$ is:
\begin{equation*}
\begin{split}
E_l (z^k) =&
\exp\Big( \frac{ \beta}{r+1}\cdot l\cdot \sum_{i=0}^{r}z^{-l}\big(z\frac{\partial}{\partial z})^iz^{l}
(z\frac{\partial}{\partial z}\big)^{r-i}\Big) (z^{k}) \\
=&\exp \Big( \frac{\beta}{r+1} \big( (z\partial_z+l)^{r+1}-(z\partial_z)^{r+1} \big)
\Big) (z^k) \\
=&\exp \Big( \frac{\beta}{r+1} \big( (k+l)^{r+1}-k^{r+1} \big)
\Big) (z^k) \\
=& e^{f(k+l)-f(k)}z^k.
\end{split}
\end{equation*}
Now the lemma follows from this relation and \eqref{eq-PQ-E}.
\end{proof}

As a consequence of the above lemma,
we may easily check that:
\begin{equation*}
\begin{split}
&PQ(z^k) = (k+1) z^k -p e^{f(k) - f(k-1)} z^{k-1};\\
&QP(z^k) = kz^k -pe^{f(k) - f(k-1)} z^{k-1},
\end{split}
\end{equation*}
and then:
\begin{equation*}
(PQ-QP)(z^k) = z^k,
\end{equation*}
which proves the commutation relation $[P,Q] =1$.

Finally,
we need to check \eqref{eq-PQ-action}.
Recall that $r$ is even,
hence $f$ is an odd function.
Then we have:
\begin{equation*}
\begin{split}
 P( \Phi_k^{(r,\vartheta)} )
 =& P\cdot \Big(z^k+\sum_{i\ge1}2(-1)^i \cdot \frac{p^{k+i}}{4\cdot k! \cdot i!}
 \cdot\frac{2i}{i+k}  e^{f(k)+f(i)}z^{-i} \Big) \\
 =&  ke^{f(k)-f(k-1)}z^{k-1}-pe^{f(k)-f(k-2)}z^{k-2}\\
  &+\sum_{i\ge1} 2(-1)^i \cdot
  \frac{p^{k+i}}{4 \cdot k ! \cdot i !} \cdot
   \frac{2i}{i+k} e^{f(k)+f(i)} \times\\
   &\qquad \Big(-i e^{f(i+1)-f(i)} z^{-i-1}
    -p e^{f(i+2)-f(i)} z^{-i-2}\Big).
 \end{split}
\end{equation*}
On the other hand,
\begin{equation*}
\begin{split}
& k e^{f(k)-f(k-1)} \Phi_{k-1}^{(r,\vartheta)}
-p e^{f(k)-f(k-2)} \Phi_{k-2}^{(r,\vartheta)} \\
=& k e^{f(k) - f(k-1)} \Big(
z^{k-1} +\sum_{i\geq 1} 2(-1)^i
\frac{p^{k-1+i}}{4\cdot (k-1)!\cdot i!} \cdot \frac{2i}{i+k-1}
e^{f(k-1) + f(i)} z^{-i} \Big) \\
& -p e^{f(k) - f(k-2)} \Big(
z^{k-2} +\sum_{i\geq 1} 2(-1)^i
\frac{p^{k-2+i}}{4\cdot (k-2)!\cdot i!} \cdot \frac{2i}{i+k-2}
e^{f(k-2) + f(i)} z^{-i} \Big).
 \end{split}
\end{equation*}
Now compare the coefficient of $z^{-i}$ in the above two expressions.
In order to prove the first equality in \eqref{eq-PQ-action},
we only need to prove the following identity:
\begin{equation*}
\begin{split}
&(-1)^i p^{k-1+i} e^{f(k)+f(i)} \Big(
\frac{k}{(k-1)! \cdot i!} \cdot \frac{i}{i+k-1} - \frac{1}{(k-2)! \cdot i!}\frac{i}{i+k-2}\Big)\\
=& (-1)^i p^{k-1+i} e^{f(k)+f(i)} \Big(
\frac{1}{k! \cdot (i-1)!} \cdot \frac{(i-1)^2}{i+k-1} - \frac{1}{k! \cdot (i-2)!}\frac{i-2}{i+k-2}\Big),
\end{split}
\end{equation*}
or equivalently,
\begin{equation*}
\begin{split}
\frac{k^2 i}{i+k-1} - \frac{k(k-1)i}{i+k-2}=
\frac{i(i-1)^2}{i+k-1} - \frac{i(i-1)(i-2)}{i+k-2},
\end{split}
\end{equation*}
which can be checked directly.
Therefore we have proved:
\begin{equation*}
P( \Phi_k^{(r,\vartheta)} )
= k e^{f(k)-f(k-1)} \Phi_{k-1}^{(r,\vartheta)}
-p e^{f(k)-f(k-2)} \Phi_{k-2}^{(r,\vartheta)},
\end{equation*}
where we have used the conventions $\Phi_{-1}^{(r,\vartheta)} = \Phi_{-2}^{(r,\vartheta)} =0$.
By taking $k=0$,
we obtain the proof of $P( \Phi_k^{(r,\vartheta)} ) =0$.
The second equality in \eqref{eq-PQ-action} can be proved using exactly the same method,
and here we omit the details.

\vspace{.2in}

{\em Acknowledgements}.
We  thank an anonymous referee for helpful suggestions.
We also thank Prof. Shuai Guo for fruitful discussions,
and thank Prof. Huijun Fan, Prof. Xiaobo Liu, and Prof. Jian Zhou for encouragement.

\end{document}